\documentclass[aps,pra,twocolumn,showpacs,floats]{revtex4-1}
\usepackage[utf8]{inputenc}
\usepackage[unicode=true, breaklinks=false, pdfborder={0 0 1}, backref=false, colorlinks=true, linkcolor=blue, urlcolor=blue, citecolor=blue]{hyperref}
\usepackage{amsmath}
\usepackage{amssymb}
\usepackage{amsthm}
\usepackage{graphicx}
\usepackage{physics}
\usepackage{bbold}
\usepackage{xcolor}
\usepackage{extarrows}

\newtheorem{theorem}{Theorem}
\graphicspath{ {Images/} }
\DeclareGraphicsExtensions{.pdf,.png}
\newcommand{\be}{\begin{equation}}
\newcommand{\ee}{\end{equation}}
\newcommand{\bea}{\begin{eqnarray}}
\newcommand{\eea}{\end{eqnarray}}
\newcommand{\ba}{\begin{array}}
\newcommand{\ea}{\end{array}}

\newcommand{\noi}{\noindent}

\newcommand{\ra}{\rangle}

\begin{document}

\title{\bf Enhanced energy transfer in a Dicke quantum battery}
\author{X. Zhang} \email{x.zhang-8@tudelft.nl}
\author{M. Blaauboer} \email{m.blaauboer@tudelft.nl}
\affiliation{Kavli Institute of Nanoscience, Delft University of Technology, Lorentzweg 1. 2628 CJ Delft, The Netherlands}

\date{\today}

\begin{abstract}
We theoretically investigate the enhancement of the charging power in a Dicke quantum battery which consists of an array of $N$ two-level systems (TLS) coupled to a single mode of cavity photons. In the limit of small $N$, we analytically solve the time evolution for the full charging process. The eigenvectors of the driving Hamiltonian are found to be pseudo-Hermite polynomials and the evolution is thus interpreted as harmonic oscillator like behaviour. We find that there exists a universal flip duration in this process, regardless to the number of TLSs inside the cavity. Then we demonstrate that the average charging power when using a collective protocol is $\sqrt{N}$ times larger than the parallel charging protocol as for transferring the same amount of energy. Unlike previous studies, we point out that such quantum advantage does not originate from entanglement but dues to the coherent cooperative interactions among the TLSs. Our results provide intuitive quantitative insight into the dynamic charging process of a Dicke battery and can be observed under realistic experimental conditions.
\end{abstract}

\maketitle

\section{Introduction}
\label{sec:introduction}
\noi Batteries have become ubiquitous in modern technology, supplying power to devices as small as nano-robots and as large as automotive engines. However, the continuing miniaturization technology gradually pushes those traditional batteries into the atomic limit, i.e. the quantum world. This trend, rather than bringing us into an uncontrollable regime, offers the possibility of utilizing quantum properties for investigating and developing more efficient energy manipulations. 

The emerging field of quantum batteries, started by Alicki and Fannes~\cite{Fannes2013}, is aimed at searching for adequate protocols based on quantum coherence and entanglement for more efficient charging-discharging energy transfer. In general, a quantum battery is a system possessing discrete energy levels and interacting with external driving and consumption sources in a controllable fashion. The internal energy of a quantum battery is defined as $\mathrm{tr}[\rho H_B]$, with $\rho$ the density matrix describing the state of the battery and $H_B$ the battery Hamiltonian (see section~\ref{sec:model_desp}). Charging (discharging) of a quantum battery means evolving into a higher (lower) energetic state $\rho'$ by cyclic unitary operations. From the viewpoint that information is a form of energy, research on quantum batteries intrinsically involves using the notions and techniques of quantum information~\cite{John2016}. Questions like whether entanglement plays a role in speeding up the energy transfer and how entropy (and related concepts) evolve in specific battery systems are under active research.~\cite{Giovannetti2003,Borras,Eduardo,Henao2018,Andolina1,Andolina2,zhang2018}

As an answer to these questions, Binder~\emph{et al.} suggested that energy can be deposited into an array of $N$ work qubits with speed-up in charging time $T$ such that $T_\mathrm{global}=T_\mathrm{par}/N$ for the use of a globally entangling charging Hamiltonian compared to a parallel individual protocol~\cite{Binder2015}. In consequence, the average charging power defined by $\expval{P}=\mathrm{tr}[(\rho'-\rho)H_B]/T$ is $N$ fold stronger for the entangling charging protocol compared to the parallel one. However, such a global entangling operation involves highly non-local interactions, which might be difficult to realize in practice. Le \emph{et al.} therefore designed a practical model consisting of a solid state spin-chain driven under experimentally available resources such as electron spin resonance and exchange interactions~\cite{Thao2018}. They predict that in the strong coupling regime the time-averaged charging power for an entangling protocol is actually worse than the individual charging. While the instantaneous charging power could be large, the total amount of energy stored in a spin chain is negligibly small. This conflicting scenario leads us to investigate if the enhancement of charging can always be attributed to the shortened passage through the entangled subspace~\cite{Thao2018}.

Another practical setting for global charging is based on the Dicke model~\cite{Dicke1954}, which describes an array of two-level systems (TLS) enclosed in a photon cavity whose frequency is on resonance with the Zeeman splitting of those TLSs. A recent paper shows that in the thermodynamic limit (i.e. the number of TLS $N\gg 1$) quantum enhancement of charging power is proportional to $\sqrt{N}$ in the normal phase and $N$ in the superradiant phase~\cite{Ferraro2018,Emary2003}. It is argued in Ref.~\onlinecite{Ferraro2018} that the cause of such enhancement is the entanglement between TLSs, which is induced by the sharing of photons in the cavity. However, the conclusion from Ref.~\onlinecite{Wolfe2014} suggests that there is actually no entanglement generated in the Dicke superradiant phase. Based on this conclusion and the conjecture of Ref.~\onlinecite{Thao2018} that entanglement may not be the only cause of quantum speed-up, we investigate in this paper the question whether for other limits of a Dicke quantum battery there is also a speed-up effect in charging and whether entanglement plays a role in this.

In particular, we analytically prove that in the opposite limit, when the number of TLS is much smaller than the number of cavity photons ($N\ll n$), there is also an enhancement of charging-discharging efficiency. As for analytic calculations we limit the Dicke Hamiltonian to a range that the coupling strength is smaller than the Zeeman splitting of the TLS (also named spins in this paper). By using the rotating wave approximation (RWA) we calculate the time $\tau$ required to charge $N$ spins from the ground state to the excited state. We show that $\tau=\pi/2g\sqrt{n}$, independent to the number of spins $N$. Given this universal flipping duration we argue that the power for a collective charging protocol is $\sqrt{N}$ times larger than for an individual charging one. Quite contrary to previous studies~\cite{Fannes2013,Binder2015,Ferraro2018}, there is no entanglement created during such a process. By solving the von Neumann equation, we clearly see that the source of speed-up comes from the coherent but non-entangling cooperative interaction among the spins.

The remaining part of this paper is organized as follows. In Sec.~\ref{sec:model_desp} we provide the definitions of fundamental concepts, charging protocols, and specify the Hamiltonian for our Dicke battery. Next in Sec.~\ref{sec:charging} we present the key part of this paper: the mathematical proof that there exists a universal spin flip duration for collective charging protocols. Based on this result we then study the boost of the energy transfer efficiency for our Dicke quantum battery. Sec.~\ref{sec:reason} is devoted to investigating the origin of such enhancement and comparing our model with other studies. We provide concluding remarks in Sec.~\ref{sec:conclusion}.

\section{Model description}
\label{sec:model_desp}
\begin{figure}
    \begin{centering}
    \includegraphics[width=\columnwidth]{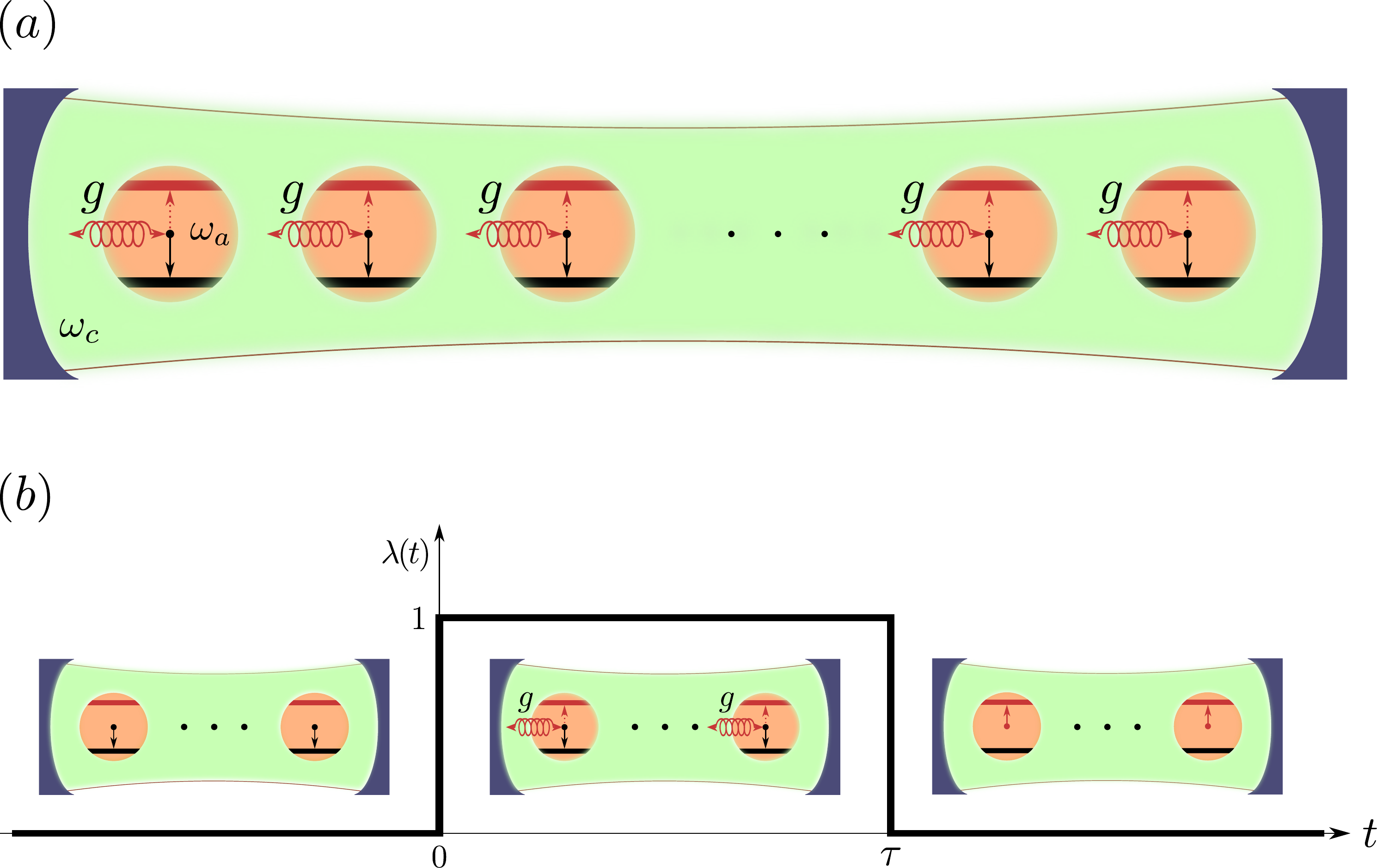}
    \par\end{centering}
    \caption{(Color online) (a) Schematic representation of a Dicke quantum battery as an array of identical two level systems with energy splitting $\hbar\omega_a$. The ground state (black bar) and excited state (red bar) are equivalent to the states of spin down and up (see arrows). The batteries are charged inside a single cavity (green background) of photonic mode $\omega_c$ and $g$ is the coupling constant among TLSs and cavity photons. (b) Initially the coupling signal $\lambda$ is set to be $0$ and the batteries are prepared in the ground state. The coupling is switched on at time $t=0^+$ and maintained as constant $g$ for a period of $\tau$. Then we turn off the coupling and the charging process therefore stops. Energy is transferred from cavity photons into the batteries who final states are expected to be fully charged (all spins up).}
    \label{fig:charging}
\end{figure}

The model of a quantum battery, shown in Fig.~\ref{fig:charging}, is an array of two level systems enclosed in a photon cavity. Since there is an equivalence between the ground/excited state of two level system to the spin down/up state of Zeeman splitting, we will refer to these TLSs as spins in the following. Without loss of generality, we set the spin down $\ket{\downarrow}$ to be the ground state and initialize the battery system into the state of all spins down $\ket{\psi_0}=\ket{\downarrow,\downarrow,\dots,\downarrow}$. When applying a magnetic filed $B_z$ the internal Hamiltonian of the battery system reads:
\begin{equation}
    H_B=g^*\mu_B B_z S_z
\label{eq:batteryHam}
\end{equation}
where $S_z=\sum_{j=1}^N \sigma_z^{(j)}$, the index $j$ refers to the $j$th spin inside the cavity and $\sigma_z$ denotes the Pauli spin operator in the $z$-direction. $S_z$ is the total spin operator, counting the Zeeman splitting for all the spins. Working in units of $\hbar=1$, the battery Hamiltonian (\ref{eq:batteryHam}) can be simplified as:
\begin{equation}
    H_B=\omega_a S_z
\end{equation}
where the frequency $\omega_a$ can be tuned by changing the external magnetic field $B_z$. The initial energy for such a $N$-spin battery then corresponds to be $E_0=\matrixel{\psi_0}{H_B}{\psi_0}=-\frac{N}{2}\omega_a$, and the energy stored in the battery is given by:
\begin{equation}
    W(t)=\matrixel{\psi_t}{H_B}{\psi_t}-\matrixel{\psi_0}{H_B}{\psi_0}.
\end{equation}
In this paper, we focus on the time required to flip all spins down $\ket{\psi_0}=\ket{\downarrow,\downarrow,\dots,\downarrow}$ to all spins up $\ket{\psi_\tau}=\ket{\uparrow,\uparrow,\dots,\uparrow}$. So the energy stored in the battery by this process is expected to be:
\begin{equation}
    W(\tau)=\matrixel{\psi_\tau}{H_B}{\psi_\tau}-\matrixel{\psi_0}{H_B}{\psi_0}=N\omega_a
    \label{eqn:energy}
\end{equation}
and the average charging power is:
\begin{equation}
    P(\tau)=\frac{W(\tau)}{\tau}.
    \label{eqn:power}
\end{equation}
The cavity, as the charger, is set to stay in a single mode of the quantized electromagnetic field. Its internal Hamiltonian reads $H_c=\omega_c a^\dagger a$, with $\omega_c$ the photon frequency and $a^\dagger, a$ the creation and annihilation operators of cavity photons. We assume the cavity to be completely isolated, i.e. there is no leakage of photons. Moreover, we assume that by tuning the external magnetic field $B_z$, the Zeeman splitting of the spins is on resonance with the cavity photons, i.e. $\omega_c=\omega_a$. Originating from magnetic dipole interactions, the coupling between the cavity photons and spins is modeled as the Dicke interaction $H_\mathrm{int}=g(a^\dagger+a)(S_+ + S_-)$, where $S_\pm=\sum_{j=1}^N \sigma_\pm^{(j)}$ are the (summed) raising and lowering operators and $g$ is the coupling constant. The full Hamiltonian describing the charger-battery system is now given by:
\begin{equation}
    H_\mathrm{Dicke}(t) =\omega_c a^\dagger a+\omega_a S_z+\lambda(t)g(a^\dagger+a)(S_++S_-).
    \label{eqn:dicke}
\end{equation}
Here $\lambda(t)$ is a time-dependent coupling signal set to be $1$ during the charging period $[0,\tau]$ and $0$ otherwise, as illustrated in Fig.~\ref{fig:charging}. Before $t=0$ there is no coupling between photons and spins, and the cavity is assumed to stay in a $n$-photon Fock state $\ket{n}$. The initial state of total system $|\psi_0\ra$ then reads:
\begin{equation}
    \ket{\Psi_0}=\ket{\downarrow,\downarrow,\dots,\downarrow}\otimes\ket{n}
    \label{eqn:initial}
\end{equation}
At $t=0^+$ the coupling is turned on. Driven by the Dicke Hamiltonian (\ref{eqn:dicke}) energy starts to be transferred from cavity photons to battery spins. The coupling is kept constant during this charging period $t\in[0,\tau]$ and shall be switched off at $t=\tau$. The quantum state of the total system will then be static again with the battery spins expected to be fully flipped up.

In typical experiment settings, the coupling constant $g$ is much smaller than cavity energy $\omega_c$ so that the Dicke interaction can be simplified using the rotating wave approximation. Thus resulting in the Tavis-Cumming Hamiltonian as from Ref.~\onlinecite{Soykal2010_letter}:
\begin{equation}
    H=\omega_c a^\dagger a+\omega_c S_z+g(S_+ a+S_- a^\dagger)
    \label{eqn:charging}
\end{equation} 
Based on this Hamiltonian and the initial state~(\ref{eqn:initial}) we can calculate the time $\tau$ used to flip spins from the down to the up state. But before that we point out four prominent properties of this charging protocol. i) When the photon frequency is on resonance with Zeeman energy, we have $[H,\,a^\dagger a+S_z]=0$, i.e. the ``excitation number" is conserved during the evolution. For example if we have 3 spins inside the cavity, the initial state can be denoted as $\ket{\frac{3}{2},-\frac{3}{2},n}$ with $J=\frac{3}{2}$ and $M=-\frac{3}{2}$ standing for the state of three spins down and $n$ the initial photon number. During the evolution this charger-battery system can only evolve into states $\ket{\frac{3}{2},-\frac{1}{2},n-1}$, $\ket{\frac{3}{2},\frac{1}{2},n-2}$, and $\ket{\frac{3}{2},\frac{3}{2},n-3}$, keeping the excitation number $n-\frac{3}{2}$ conserved. ii) The coupling term in Eq.~(\ref{eqn:charging}) commutes with the remaining part of Hamiltonian $H$, meaning that there is no thermodynamic work cost for switching on and off the coupling~\cite{Andolina1}. iii) The Tavis-Cummings Hamiltonian is not exactly solvable except for two limiting cases. The first one is when the number of spins $N$ approaches thermodynamic limit. In this situation one can transform the collective spin operators using Holstein-Primakoff transformation which leads to a simplified Hamiltonian of quadratic form~\cite{Primakoff1940}. A second case is for the number of spins $N$ much less than the number of photons, $N\ll n$. As studied in this paper, the second case will result in solving a symmetric tridiagonal matrix and this can be done analytically. iv) The initial Fock state is relatively difficult to prepare in practice. But according to the conclusion of Ref.~\onlinecite{Andolina1}, the same efficiency of energy transfer can be achieved by replacing the Fock state with a coherent state of the same energy.

\section{Collective Charging of a Dicke Quantum Battery}
\label{sec:charging}
\subsection{Matrix representation}
In order to calculate the time of flipping spins from down to up, we shall first derive the matrix representation of the charging Hamiltonian~(\ref{eqn:charging}). Due to the conservation of excitation number, the quantum states of such charger-battery system flipping from $\ket{\frac{N}{2},-\frac{N}{2},n}$ to $\ket{\frac{N}{2},\frac{N}{2},n-N}$ can only evolve within the subspace in which the total spin momentum ($J=\frac{N}{2}$) is preserved. Thus the matrix representation of Hamiltonian is limited to dimension $N+1$ instead of $2^N$. Therefore we can identify those quantum states by the ($N+1$)-dimensional basis vectors, such as:
\begin{align}
    \begin{split}
    &\ket{\frac{N}{2},-\frac{N}{2},n}=\begin{pmatrix}0&\cdots&0&1\end{pmatrix}^T\\
    &\ket{\frac{N}{2},-\frac{N}{2}+1,n-1}=\begin{pmatrix}0&\cdots&1&0\end{pmatrix}^T\\
    &\qquad\vdots\\
    &\ket{\frac{N}{2},\frac{N}{2},n-N}=\begin{pmatrix}1&\cdots&0&0\end{pmatrix}^T
    \end{split}
\end{align}
With this identification, the representation of such a $3$-spins battery is $4\times 4$ matrix:
\begin{equation*}
    H=\begin{pmatrix}
    (n-\frac{3}{2}) & g\sqrt{3(n-3)} & 0 & 0\\
    g\sqrt{3(n-3)} & (n-\frac{3}{2}) & g\,2\sqrt{n-1} & 0\\
    0 & g\,2\sqrt{n-1} & (n-\frac{3}{2}) & g\sqrt{3n}\\
    0 & 0 & g\sqrt{3n} & (n-\frac{3}{2}).
    \end{pmatrix}
\end{equation*}
In practice the number of photons $n$ is usually much larger than $1$, so we can factorize this matrix into two parts:
\begin{equation*}
\small
    H=(n-\frac{3}{2})\mathbb{1}+g\sqrt{n}\begin{pmatrix}0&\sqrt{3}\sqrt{1}&0&0\\\sqrt{3}\sqrt{1}&0&\sqrt{2}\sqrt{2}&0\\0&\sqrt{2}\sqrt{2}&0&\sqrt{1}\sqrt{3}\\0&0&\sqrt{1}\sqrt{3}&0\end{pmatrix}.
\end{equation*}
Actually it is a common feature that in the limit of large photon number, we can always approximate the matrix representation of the charging Hamiltonian by the form $H=\mathrm{\#Ex}\,\mathbb{1}+\tilde{H}$, with $\mathrm{\#Ex}$ the excitation number and
\begin{equation}
    \tilde{H}=g\sqrt{n}
    \begin{pmatrix}
    0 & b_1 \\
    b_1 & 0 & b_2 \\
    & b_2 & 0 & b_3\\
    & & \ddots & \ddots & \ddots\\
    & & & b_{N-1} & 0 & b_N \\
    & & & & b_N & 0
    \end{pmatrix}
    \label{eqn:hamiltonian}
\end{equation}
with each element $b_{k}=\sqrt{N-k+1}\sqrt{k}$ and $k=1,2,\dots,N$. Because the identity matrix commutes with all other operators, the first term of $H$ only adds a common phase factor to the quantum state which does not influence the spin-flip duration. So we can ignore the first term and only include the second part of $H$ (i.e. the symmetric tridiagonal matrix) in the calculation of propagator $U(t)=e^{-i\tilde{H}t}$.

Given the matrix form of $\tilde{H}$ we are now ready to determine analytically the flip duration $\tau$ for fully charging the spins:
\begin{equation}
    \begin{pmatrix}1\\0\\\vdots\\0\end{pmatrix}\xlongequal{\tau?}e^{-i\tilde{H}\tau}\begin{pmatrix}0\\0\\\vdots\\1\end{pmatrix}.
    \label{eqn:flipping}
\end{equation}
However, direct substitution of the matrix expression [Eq.~(\ref{eqn:hamiltonian})] into the propagator $e^{-i\tilde{H}t}$ results in a cumbersome expression. In practice, we therefore first diagonalize this Hamiltonian such that $\tilde{H}=VDV^\dagger$, with $D$ a diagonal matrix whose elements represent the eigenvalues of $\tilde{H}$, and $V$ a unitary matrix in which each column stands for the corresponding eigenvectors. After that we can write the propagator as $U(t)=Ve^{-iDt}V^\dagger$.

\subsection{General form of eigenvalues and eigenvectors of the Tavis-Cumming Hamiltonian}
The calculation of eigenvalues and eigenvectors of matrix $\tilde{H}$~(\ref{eqn:hamiltonian}) is straightforward but involves rather technical expressions. Here we therefore only present the final results. The eigenvalues of $\Tilde{H}$ are:
\begin{equation}
    D=g\sqrt{n}\begin{pmatrix}N&&&&\\&N-2&&&\\&&\ddots&&\\&&&-N+2\\&&&&-N\end{pmatrix}.
    \label{eqn:eigenvalue}
\end{equation}
Disregarding the common factor $g\sqrt{n}$, these eigenvalues form a sequence from $N$ to $-N$ with a consecutive difference of $-2$. These results coincide with the numerical analysis of the Rabi splitting for Tavis-Cumming Hamiltonian in the large photon limit~\cite{Chiorescu2010}. The eigenvectors corresponding to eigenvalues $g\sqrt{n}(N-2k)\;(k=0,1,\dots,N)$ are:
\begin{equation}
    V^k=\frac{1}{2^N}\frac{1}{k!\,\sqrt{\binom{N}{k}}}\begin{pmatrix}x_0^k\\x_1^k\\\vdots\\x_N^k\end{pmatrix},
    \label{eqn:eigenvector}
\end{equation}
with each of the column entries $x_\xi^k\;(\xi=0,1,\dots,N)$ given by:
\begin{equation}
    x_\xi^k=\sqrt{\binom{N}{\xi}}P_k(\xi).
    \label{eqn:eigen_x}
\end{equation}
The characteristic polynomials $P_k(\xi)$ obey the recursion relation:
\begin{align}
\begin{split}
    &P_0(\xi)=1\\
    &P_1(\xi)=N-2\xi\\
    &P_k(\xi)=[N-2\xi]P_{k-1}-(k-1)(N-k+2)P_{k-2}.
    \label{eqn:recursion}
\end{split}
\end{align}
They are orthogonal to each other and alternating between even and odd parity. Actually, as the number of spins $N$ approaches the thermodynamic limit, these polynomials converge into the Hermite polynomials. Because $\xi=0,1,\dots,N$, we have $N-2\xi=N,N-1,\dots,-N$. If we let $N-2\xi=Nx$, then $x$ goes from $1$ to $-1$ in discrete steps and the recursion relation (\ref{eqn:recursion}) reads:
\begin{align}
\begin{split}
    &P_0(x)=1\\
    &P_1(x)=Nx\\
    &P_k(x)=NxP_{k-1}-N(k-1)P_{k-2}
    \label{eqn:pseudo}
\end{split}
\end{align}
with corresponding Rodrigues' formula:
\begin{equation}
    P_k(x)=(-1)^ke^\frac{Nx^2}{2}\frac{d^k}{dx^k}e^{-\frac{Nx^2}{2}}.
\end{equation}
Eq.~(\ref{eqn:pseudo}) is just the scaled version of the standard Hermite polynomials whose Rodrigues' formula is $H_k(x)=(-1)^ke^{x^2}\frac{d^k}{dx^k}e^{-x^2}$. This phenomenon can be well understood because Hermite polynomials represent the eigenstates of a quantum harmonic oscillator (QHO)~\cite{Shankar}. For an array with finite number of spins, however, the spectrum resembles that of a pseudo-harmonic oscillator with equally spaced eigenvalues and corresponding eigenvectors [Eq.~(\ref{eqn:eigenvector})] of the pseudo-Hermite form. 
One should notice that for QHO the coordinate $x$ in eigenvectors takes the form of a continuous variable in real space, while the coordinate $x$ in expression~(\ref{eqn:pseudo}) takes the form as discrete variables in the ($N+1$)-dimensional spin space. Thus it could be understood that if the number of spins stays finite, the battery behaves like a pseudo-harmonic oscillator swinging in quantized spin-momentum space.

Another feature of this recursion formula is that the second factor $(k-1)(N-k+2)$ for polynomial $P_{k}(\xi)$ exactly equals the square of the entry in $k$th row of matrix $\tilde{H}$, keeping the sum $(k-1)+(N-k+2)=N+1$ equals to the dimension of total spin space.

\subsection{The universal flip duration}
Equipped with the eigenvalues $D$~(\ref{eqn:eigenvalue}) and eigenvectors $V$~(\ref{eqn:eigenvector}), we can now replace the propagator $e^{-i\tilde{H}\tau}$ in Eq.~(\ref{eqn:flipping}) with $Ve^{-iD\tau}V^\dagger$. The time required for fully charging $N$ spins from the ground state to the excited state is summarized as theorem~\ref{thm:universal}.
\begin{theorem}
In the limit of $N/n\ll 1$, the time to flip all spins from down to up equals $\tau=\frac{\pi}{2g\sqrt{n}}$, independent to the number of spins $N$.
\label{thm:universal}
\end{theorem}
\begin{proof}
Substituting the expression of $D$ and $V$ one finds the matrix equation:
\begin{equation}
    \begin{pmatrix}1\\0\\\vdots\\0\end{pmatrix}\xlongequal{\tau?}Ve^{-iD\tau}V^\dagger\begin{pmatrix}0\\0\\\vdots\\1\end{pmatrix}.
    \label{eqn:flipping2}
\end{equation}
can be rephrased into two algebraic equations with even and odd number of spins respectively:

\textbf{Case I}: For odd numbers of spins $N=2m+1$, the algebraic equation corresponding to Eq.~(\ref{eqn:flipping2}) is
\begin{equation}
    \sum_{k=0}^m \frac{2\,\binom{2m+1}{k}}{2^{2m+1}}(-1)^k\sin{(2m+1-2k)g\sqrt{n}\tau}\xlongequal{\tau?}\pm 1.
    \label{eqn:algebraic_odd}
\end{equation}
One can see that if $\tau=\frac{\pi}{2g\sqrt{n}}$, then
\begin{equation*}
(-1)^k\sin{(2m+1-2k)g\sqrt{n}\tau}=(-1)^m
\end{equation*}
holds such that the left-hand side of (\ref{eqn:algebraic_odd}) becomes:
\begin{equation*}
    (-1)^m\sum_{k=0}^m \frac{2\,\binom{2m+1}{k}}{2^{2m+1}}=(-1)^m=\pm 1.
\end{equation*}
So in case that $N$ is odd, the flip duration $\tau=\frac{\pi}{2g\sqrt{n}}$.

\textbf{Case II}: For even number of spins $N=2m$, the algebraic equation corresponding to Eq.~(\ref{eqn:flipping2}) is given by
\begin{equation}
    \sum_{k=0}^{m-1}\frac{2\,\binom{2m}{k}}{2^{2m}}(-1)^k\cos{(2m-2k)g\sqrt{n}\tau}+(-1)^m\frac{\binom{2m}{m}}{2^{2m}}\xlongequal{\tau?}\pm 1.
    \label{eqn:algebraic_even}
\end{equation}
Similarly as the odd-spin case above, for $\tau=\frac{\pi}{2g\sqrt{n}}$
\begin{equation*}
(-1)^k\cos{(2m-2k)g\sqrt{n}\tau}=(-1)^m
\end{equation*}
holds such that the left-hand side of (\ref{eqn:algebraic_even}) becomes:
\begin{align*}
    (-1)^m\sum_{k=0}^{m-1}\frac{2\,\binom{2m}{k}}{2^{2m}}&+(-1)^m\frac{\binom{2m}{m}}{2^{2m}}=(-1)^m=\pm 1
\end{align*}
Thus in case that $N$ is even, the flip duration is given by $\tau=\frac{\pi}{2g\sqrt{n}}$. 

We conclude that the time required to flip arbitrary number of spins equals $\frac{\pi}{2g\sqrt{n}}$. Since this result is derived based on the Tavis-Cummings charging Hamiltonian, it is valid under the assumption that the number of photons saturates the number of spins~\cite{Canming2017}.
\end{proof}

\subsection{Quantum speed up in the collective charging protocol}
Based on theorem 1, we now prove that from the aspect of energy transfer there is an enhancement for the charging efficiency in aforementioned charger-battery setup.

\begin{theorem}
If the time of flipping $N$ spins equals $\tau=\frac{\pi}{2g\sqrt{n}}$, then the average charging power for collective charging protocol is $\sqrt{N}$ times larger than for the corresponding parallel (individually charging) protocol.
\label{thm:speedup}
\end{theorem}
\begin{figure}
    \centering
    \includegraphics[width=\columnwidth]{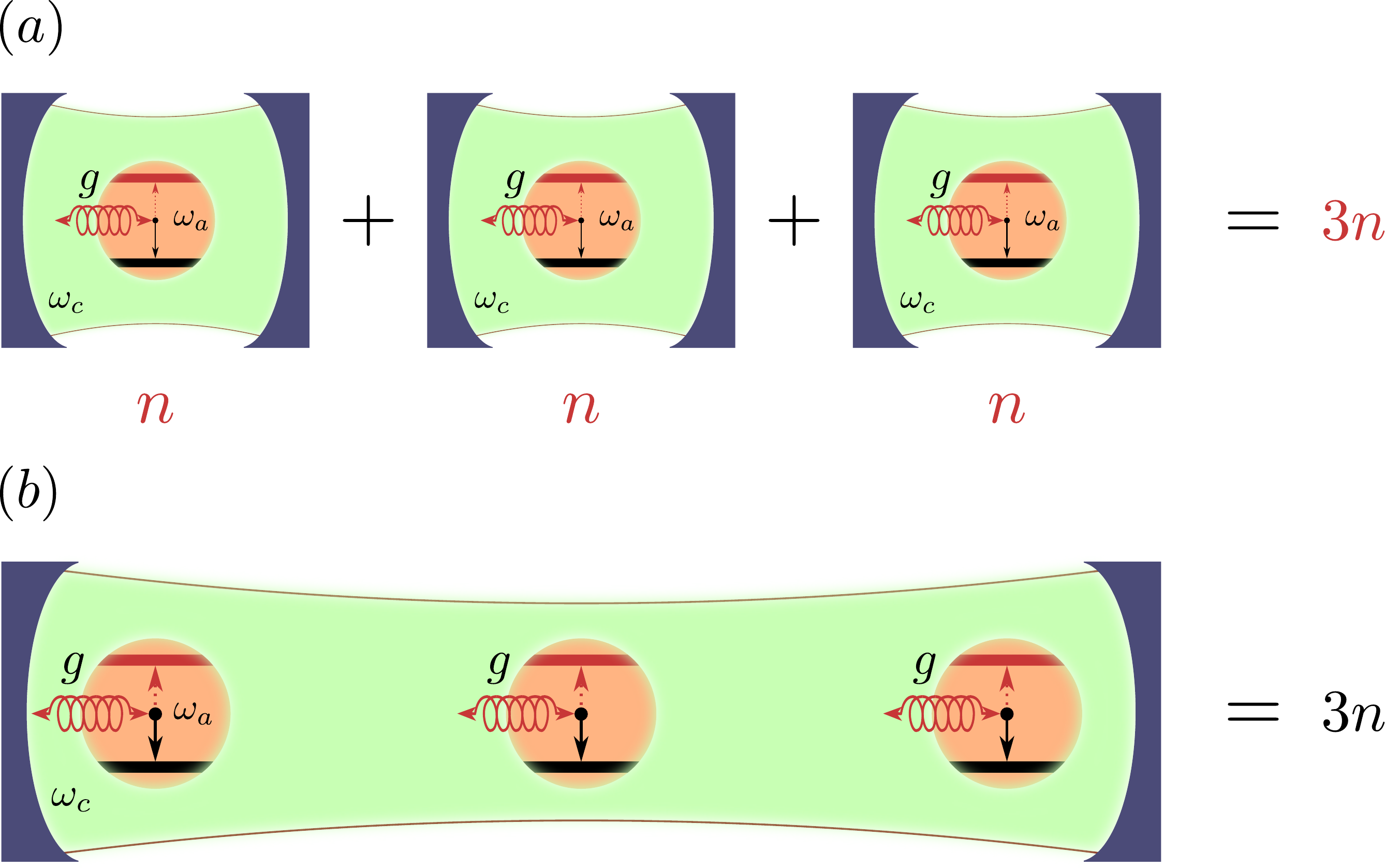}
    \caption{(Color online) (a) Individually charging three spins in parallel protocol. Each cavity is filled with $n$ photons and in total $3n$ number of photons are used. (b) The three spins are collectively charged inside a single cavity. To make a fair comparison the same amount of photons should be provided as in the parallel protocol.}
    \label{fig:indcol}
\end{figure}
\begin{proof}
Let us set $N=3$ as an example. The proof can be straightforwardly extended to arbitrary $N$ spins. First, it is easy to see that the time of flipping three spins in parallel is the same as the time required to flip each spin individually. That is, $\tau_\mathrm{par}=\frac{\pi}{2g\sqrt{n}}$ where $n$ denotes the number of driving photons in each cavity (as shown in figure~\ref{fig:indcol}(a)). Therefore, the total amount of photons for driving $3$ spins simultaneously should be summed to be $3n$. In order to make a fair comparison, it is important to make sure that the energy of the charging field in the collective protocol (shown in figure~\ref{fig:indcol}(b)) equals to the parallel one. That is, in this example the number of photons in collectively charging cavity should set to be $3n$ (instead of $n$). Thus the time required to flip three spins in a single cavity is given by $\tau_\mathrm{col}=\frac{\pi}{2g\sqrt{3n}}=\frac{1}{\sqrt{3}}\tau_\mathrm{par}$.

In both protocols, the energy transferred to the $3$-spin battery is the same [Eq.~(\ref{eqn:energy})]; only the time of charging collectively is $\sqrt{3}$ times faster than the parallel one. As a result, the average charging power for collective protocol is $\sqrt{3}$ times stronger than the corresponding parallel protocol. One can easily extend this result to a $N$-spin battery as long as the assumption $N\ll n$ is preserved.
\end{proof}

\section{Quantum speed up originating from coherent cooperative interactions}
\label{sec:reason}
In order to understand the origin of the speedup effect described in previous section, we first calculate the quantum speed limit (QSL) which forms the lower bound of the evolution duration that could possibly be achieved by corresponding Hamiltonian. For the parallel protocol, its energy variance reads $\Delta\tilde{H}_\mathrm{par}=N\cdot 2g\sqrt{n}$ and the number of charging photons is $nN$. According to Refs.~\onlinecite{Binder2015} and \onlinecite{Giovannetti2002}, we have the QSL as:
\begin{equation}
    \mathfrak{T}_\mathrm{par}=\frac{\pi}{2\Delta\tilde{H}_\mathrm{par}}=\frac{\pi}{4Ng\sqrt{n}}
\end{equation}
For collective charging, the energy variance for the same number of charging photons is $\Delta\tilde{H}_\mathrm{col}=N\cdot 2g\sqrt{nN}$, and the corresponding QSL is:
\begin{equation}
    \mathfrak{T}_\mathrm{col}=\frac{\pi}{2\Delta\tilde{H}_\mathrm{col}}=\frac{\pi}{4Ng\sqrt{nN}}=\frac{1}{\sqrt{N}}\mathfrak{T}_\mathrm{par}
\end{equation}
We see that due to the collective effect, the quantum speed limit has also been pushed down by a factor of $\frac{1}{\sqrt{N}}$. This analysis agrees with the discussion in Ref.~\onlinecite{Binder2016}, stating that the speed limit for parallel driving is $\sqrt{N}$ times larger than that for the collective one. Together with the previous result of $\tau_\mathrm{col}=\frac{1}{\sqrt{N}}\tau_\mathrm{par}$, we conjecture that the coherent cooperative effect among $N$ spins inside the cavity leads to a shortcut of duration by factor $\frac{1}{\sqrt{N}}$ for all time related phenomena.

By taking partial trace of the density matrix of the collectively charged spins, we can calculate and plot the development of entanglement during the evolution. As shown in Fig.~\ref{fig:flipping}, no entanglement develops during the flipping process. Upon detailed numerical inspection, one finds that all spins process exactly in step, i.e. follow the same evolution. This result differs from the conclusion of Ref.~\onlinecite{Ferraro2018}, which states that long-range entangling interactions among the spins will be formed due to the mediation of the common photon field inside the cavity. We have thus found an example where it is not the globally entangling operations that lead to the enhancement of charging. This raises the question what is the source of quantum speed-up in our model, if there is no entanglement involved.
\begin{figure}
    \centering
    \includegraphics[width=\columnwidth]{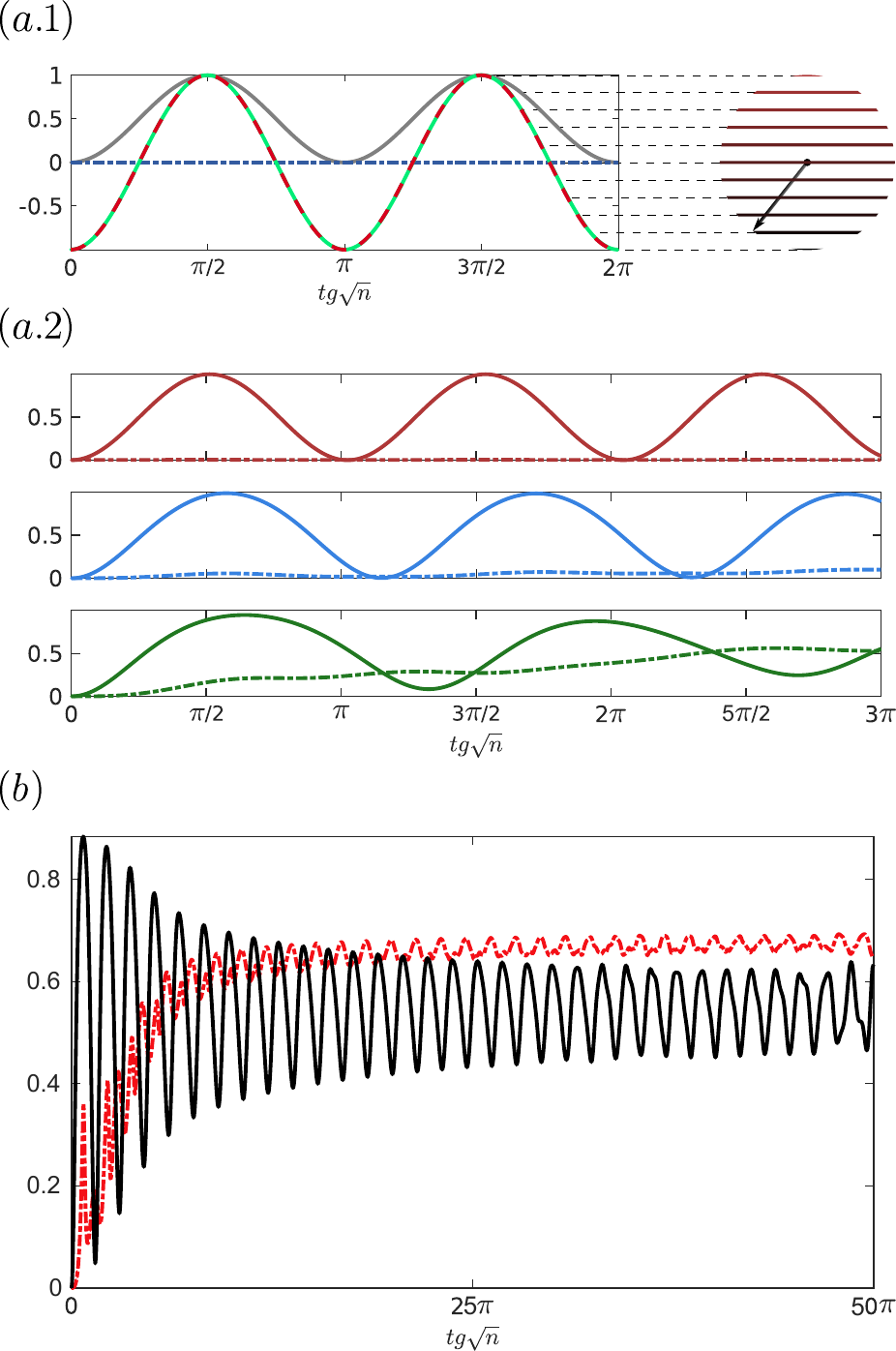}
    \caption{(Color online) (a.1) The evolution of $10$ spins with $10^4$ driving photons under the charging Hamiltonian~Eq.~(\ref{eqn:charging}). The gray curve stands for the energy deposited into each battery $\frac{W(t)}{N\omega_a}$ (in dimensionless units), which evolves periodically as the spins are consistently flipped from down to up and vice-versa. We take the Von Neumann entropy of the first spin as a measure for the amount of entanglement of the batteries. This indicator (blue dotted line) stays $0$ meaning that there is no entanglement being created during the process. The overlapping green and red curves depict the evolution of $\cos\theta_j$, with $\theta_j$ the angle between $\expval{\Vec{s}_j}$ and the $z^+$ axis for the $j$th spin (shown as the right sphere). Here the red curve stands for the first spin, i.e. $j=1$ and the green one represents for another randomly picked spin out of the batteries. This overlap means all spins in the cavity evolve exactly in step. (a.2) From top to bottom these three subplots indicate the charging of $10$ spins with $100$, $20$, and $12$ photons respectively. As before, solid curves refer to the deposited energy on each spin and the dotted line stands for the amount of entanglement. As the number of cavity photons shrinks, entanglement shows up and the spins cannot be fully charged. Such a result suggests that the requirement of $N/n\ll1$ can be reached as long as the photons outnumber the spins by an order of magnitude. (b) The supplementary plot for the evolution of $100$ spins driven by $90$ photons. As time progresses, dynamical equilibrium between the photons and spins will be formed as indicated in Ref.~\onlinecite{Soykal2010_prb}.}
    \label{fig:flipping}
\end{figure}

To answer this question we symbolically solve the von Neumann equation with Hamiltonian Eq.~(\ref{eqn:charging}):
\begin{equation}
\begin{split}
&\dv{a}{t}=i[H,a]=-i\omega_c a-igS_-\\
&\dv{\sigma_z^{(j)}}{t}=i[H,\sigma_z^{(j)}]=ig(a^\dagger\sigma_-^{(j)}-a\sigma_+^{(j)}).
\end{split}
\label{eqn:neumann}
\end{equation}
Here $j=1,2,\dots,N$ numbers the spins inside the cavity. Solving the dynamic equation for operator $a(t)$ one finds that:
\begin{equation}
\begin{split}
    &a(t)=-i g\int_{-\infty}^t e^{i\omega_c(t'-t)}S_-(t')dt'\\
    &a^\dagger(t)=i g\int_{-\infty}^t e^{-i\omega_c(t'-t)}S_+(t')dt'.
\end{split}
\label{eqn:dynamics_photon}
\end{equation}
Substituting these expressions back into Eq.~(\ref{eqn:neumann}) we obtain the equation of motion for a single spin:
\begin{equation}
\begin{split}
    &\dv{\sigma_z^{(j)}}{t}=-g^2\Bigl[\sigma_-^{(j)}(t)\int_{-\infty}^t e^{-i\omega_c(t'-t)}S_+(t')dt'\\
    &\hspace{5em}+\sigma_+^{(j)}(t)\int_{-\infty}^t e^{i\omega_c(t'-t)}S_-(t')dt'\Bigr].
\end{split}
\label{eqn:dynamics_spin}
\end{equation}
Eq.~(\ref{eqn:dynamics_spin}) shows that by integrating out the photon field, the effective force applied on an arbitrary spin $j$ is proportional to $\int_{-\infty}^t e^{i\omega_c(\tau-t)}S_\pm(\tau)d\tau$. We thus see that the interactions between the spins mediated by cavity photons act cooperatively, leading to an evenly distributed enhancement of the driving force on each of the battery spins. Since numerical calculation shows that all spins follow the same evolution in time, we can replace $S_\pm$ in Eq.~(\ref{eqn:dynamics_spin}) with $S_\pm=\sum_{j=1}^N \sigma_\pm^{(j)}=N\sigma_\pm$ leading to the spin dynamics:
\begin{equation}
\begin{split}
&\dv{\sigma_z}{t}=-Ng^2\Bigl[\sigma_-(t)\int_{-\infty}^t e^{-i\omega_c(t'-t)}\sigma_+(t')dt'\\
&\hspace{6em}+\sigma_+(t)\int_{-\infty}^t e^{i\omega_c(t'-t)}\sigma_-(t')dt'\Bigr].
\end{split}
\label{eqn:dynamics_collective}
\end{equation}
This final form explicitly shows that the coupling strength of each spin has been increased by $\sqrt{N}$ times. With this in mind, it is easy to understand that all the time related phenomena that depends linearly on the coupling strength would be accelerated by a factor of $\sqrt{N}$.

As argued by Binder~\emph{et al.} quantum speed up originates from two different sources. One is the reduction of path length between initial and final state in projected Hilbert space $\mathcal{L}(\ket{\psi_0},\ket{\psi_\tau})$ following the geodesic curve. Another is the enhanced driving energy felt by each local spin. Focussing on the fact that all spins follow the same evolution as if they are charged individually, we realize that the path length of evolution is the same for both protocols. This means that the collective protocol which does not create entanglement among spins also cannot drive them through the geodesic in projected Hilbert space~\cite{Carlini2006}. This explains that the speed up we observe in this paper is $\sqrt{N}$-fold instead of $N$-fold because the Tavis-Cumming Hamiltonian only increases the energy per spin but does not shorten the length of passage in the projected Hilbert space.

\section{Conclusion}
\label{sec:conclusion}
We have studied the energy transfer efficiency of an ideal Dicke quantum battery within the limit $N/n\ll 1$. Under the constraint of full charging, we predict a $\sqrt{N}$-fold boost of the charging power for a collective protocol compared to the parallel one. Using the matrix representation of the driving Hamiltonian we have analytically solved the eigenenergies and eigenstates for this charger-battery system. We find that the collective dynamics of spins mimics the process of swinging a pseudo-pendulum in quantized and finite-dimensional spin momentum space. We then apply these tools to the unitary evolution equation of the spins and demonstrate the existence of a universal flipping time for an arbitrary number of spins. Based on this, we show a boost of the averaged charging power for a collective charging protocol with evenly distributed driving forces.

Contrary to previous studies (see e.g. Refs.~\onlinecite{Binder2015,Ferraro2018}) which require multi-particle entanglement as the key part of quantum speed up, in our model there is no entanglement generated for the collective charging protocol. However, it is the coherent cooperative interactions inside the cavity that lead to increased coupling strength for each spin. Such effect results in the lowering of the quantum speed limit by a factor $\frac{1}{\sqrt{N}}$. In conclusion, for our Dicke battery the boost of charging power arises from the increased driving forces exerted on each spin, and not from a shortened path length in the projected Hilbert space.

Although the Dicke battery presented here only shows `half' the amount of speed up (factor $\sqrt{N}$ instead of $N$), it exhibits a scalable quantum advantage for faster energy transfer. Moreover, by cutting-edge development of cavity spintronics this battery setup is ready to be implemented in practice. For instance, nitrogen-vacancy (NV) centers are well suited for studying spin dynamics. NV spins can be optically implanted, initialized, and read out~\cite{Hanson2010}. As shown in Ref.~\onlinecite{Schroder2017}, the implantation of a few NV spins into the L3 photonic crystal cavity for coherent manipulations has already been realized. And the coupling of an ensemble of NV spins to a frequency tunable superconductor resonator has also been reported~\cite{kubo2010}.

Future work could focus on adding entangling interactions between the spins in order to further explore the remaining $\sqrt{N}$ factor of speedup. Another interesting line of research is to include spin decay and the injection of photons, which leads to a non-Hermitian quantum mechanical system~\cite{Bender2002}. By careful tuning such that the rate of injecting photons matches the decay rate of spins, this system allows for a transition into the regime of parity and time reversal symmetry, i.e. $\mathcal{PT}$ symmetry~\cite{Liu2017,Bender1998}. Within $\mathcal{PT}$ symmetry, the eigenvalues of the non-Hermitian Hamiltonian become real and several exceptional properties such as the increase of coupling strength and the decrease of the quantum speed limit would be expected~\cite{Dengke2017,Bender2007}.
\begin{acknowledgements}
We gratefully thank T. Yu and S. Dam for insightful discussions. X. Zhang acknowledge the financial support from the China Scholarship Council. This work is part of the research programme of the Foundation for Fundamental Research on Matter (FOM).
\end{acknowledgements}
%
%
\bibliographystyle{apsrev4-1}
\bibliography{battery}

\end{document}